\documentclass[runningheads]{llncs}
\usepackage{graphicx}
\usepackage{amsmath}
\usepackage{enumitem}
\usepackage[T1]{fontenc}
\def\Cal{\mathcal}

\newcommand{\cut}{{\mathop{\rm cutset}}}
\newcommand{\dif}{{\mathop{\rm dif}}}
\newcommand{\ch}{{\mathop{c}}}

\newtheorem{claimm}{Claim}

\newcommand{\maybeqed}{\qed}

\begin{document}

 \mainmatter
 
 \title{On the hardness of switching to a small number of edges}
 \author{V\' it Jel\' inek\inst{1}\thanks{Supported by CE-ITI project GACR
P202/12/G061.}
 \and Eva Jel\'\i nkov\' a\inst{2}
 \and Jan Kratochv\'\i l\inst{2}${}^\star$}
 
 \institute{Computer Science Institute\\
 and
 \and
 Department of Applied Mathematics\\
 Faculty of Mathematics and Physics, Charles University\\ 
 Malostransk\' e n\' am. 25, 118~00 Praha, Czech Republic\\
 \email{jelinek@iuuk.mff.cuni.cz, \char123 eva,honza\char125@kam.mff.cuni.cz}}
\authorrunning{V. Jel\' inek, E. Jel\'\i nkov\' a, and J. Kratochv\'\i l}

\maketitle

\begin{abstract}
Seidel's switching is a graph operation which makes a
given vertex adjacent to precisely those vertices to which it was non-adjacent
before, while keeping the rest of the graph unchanged. Two graphs are called
switching-equivalent if one can be made isomorphic to the other one by a sequence
of switches.

Jel{\'\i}nkov{\' a} et al. [DMTCS 13, no. 2, 2011]
presented a proof that it is NP-complete to decide if the input graph can be switched to contain
at most a given number of edges. There turns out to be a flaw in their proof. We
present a correct proof.

Furthermore, we prove that the problem remains NP-complete even when restricted
to graphs whose density is bounded from above by an arbitrary fixed constant.
This partially answers a question of Matou\v{s}ek and Wagner [Discrete Comput. Geom. 52,
no.~1, 2014].



\end{abstract}

\section{Introduction} \label{section:intro}

Seidel's switching is a graph operation which makes a
given vertex adjacent to precisely those vertices to which it was non-adjacent
before, while keeping the rest of the graph unchanged. Two graphs are called
switching-equivalent if one can be made isomorphic to the other one by a sequence
of switches. The class of graphs that are pairwise switching-equivalent is
called a switching class.

Hage in his PhD thesis~\cite[p.~115, Problem~8.5]{Hage01} posed the problem to characterize
the graphs that have the maximum (or minimum) number of edges in their switching class.
We call such graphs \emph{switching-maximal} and \emph{switching-minimal}, respectively.

Some properties of switching-maximal graphs were studied
by Kozerenko~\cite{Kozerenko15}. He proved that any graph with sufficiently
large minimum degree is switching-maximal, and that the join of certain graphs
is switching-maximal.
Further, he gave a characterization of triangle-free switching-maximal graphs
and of non-hamiltonian switching-maximal graphs.

It is easy to observe that a graph is switching-maximal if and only if its
complement is switching-minimal. We call the problem to decide if a graph is switching-minimal
\textsc{Switch-Minimal}.

Jel{\'\i}nkov{\' a} et al.~\cite{JSHK11-param} studied the more general problem
\textsc{Switch-Few-Edges} -- the problem of deciding
if a graph can be switched to contain at most a certain number of edges.
They presented a proof that the problem is NP-complete. Unfortunately, their proof is not correct.
Specifically, Lemma 4.3 of~\cite{JSHK11-param}, which claims to establish a
reduction from the classical \textsc{Max-Cut} problem to
\textsc{Switch-Few-Edges}, is false. The claim of the lemma fails, e.~g., on a
graph $G$ formed by two disjoint cliques of the same size.

In this paper, we provide a different proof of the NP-hardness of
\textsc{Switch-Few-Edges}, based on a reduction from a restricted
version of \textsc{Max-Cut}. 
Furthermore, we strengthen this result by proving that for any $c > 0$, 
\textsc{Switch-Few-Edges} is NP-complete even
if we require that the input graph has density at most~$c$.
We also prove that if the problem \textsc{Switch-Minimal} is co-NP-complete, then for
any $c > 0$, the problem is co-NP-complete even on graphs with density at most~$c$.

We thus partially answer a question of Matou\v sek and
Wagner~\cite{MW14} posed in connection with properties of simplicial complexes
-- they asked if deciding switching-minimality was easy for graphs of bounded density.
Our results also indicate that it might be unlikely to get an easy
characterization of switching-minimal (or switching-maximal) graphs, which
contributes to understanding Hage's question~\cite{Hage01}.



\subsection{Formal definitions and previous results}
\label{subs:form}

Let $G$ be a graph. Then the \emph{Seidel's switch of a vertex subset} $A
\subseteq V(G)$ is denoted by $S(G,A)$ and is defined by
\[S(G,A) = (V(G), E(G) \bigtriangleup \{ xy : x \in A,\ y \in V(G) \setminus A
\} ).
\]

\noindent
It is the graph obtained from $G$ by consecutive switching of the vertices of $A$.

We say that two graphs $G$ and $H$ are \emph{switching-equivalent} (denoted by
$G \sim H$) if there is a
set $A \subseteq V(G)$ such that $S(G,A)$ is isomorphic to $H$.
The set
$[G] = \{ S(G,A) : A \subseteq V(G) \}$ is called the \emph{switching class}
of $G$.

We say that a graph $G$ is \emph{$(\leq k)$-switchable} if there is a set $A
\subseteq V(G)$ such that $S(G,A)$ contains at most $k$ edges. Analogously, a
graph $G$ is \emph{$(\geq k)$-switchable} if there is a set $A
\subseteq V(G)$ such that $S(G,A)$ contains at least $k$ edges.

It is easy to observe that a graph $G$ is $(\leq k)$-switchable if and only if
the complement $\overline{G}$ is $\left(\geq \binom{n}{2} -
k\right)$-switchable.
We may, therefore, focus on $(\leq k)$-switchability only.

We examine the following problems.

\newdimen\sirka
\sirka=-\fboxsep \advance\sirka by -\fboxrule \multiply\sirka by 2
\advance\sirka by \hsize
\vskip\fboxsep
\vskip\fboxsep
\noindent
\framebox[\hsize][l]{\begin{minipage}{\sirka}
\textsc{Switch-Few-Edges}\\
 \textbf{Input:} A graph $G=(V,E)$, an integer $k$\\
 \textbf{Question:} Is $G$ $(\leq k)$-switchable?
 \end{minipage}}
\vskip\fboxsep

\vskip\fboxsep
\noindent
\framebox[\hsize][l]{\begin{minipage}{\sirka}
\textsc{Switch-Minimal}\\
 \textbf{Input:} A graph $G=(V,E)$\\
 \textbf{Question:} Is $G$ switching-minimal?
 \end{minipage}}
\vskip\fboxsep
\vskip\fboxsep

We say that a graph is \emph{switching-reducible} if $G$ is \emph{not}
switching-minimal, in other words, if there is a set $A
\subseteq V(G)$ such that $S(G,A)$ contains fewer edges than~$G$.
For further convenience, we also define the problem \textsc{Switch-Reducible}.

\vskip\fboxsep
\vskip\fboxsep
\noindent
\framebox[\hsize][l]{\begin{minipage}{\sirka}
\textsc{Switch-Reducible}\\
 \textbf{Input:} A graph $G=(V,E)$\\
 \textbf{Question:} Is $G$ switching-reducible?
 \end{minipage}}
\vskip\fboxsep
\vskip\fboxsep

Let $G = (V,E)$ be a graph. We say that a partition $V_1$, $V_2$ of $V$ is a
\emph{cut} of $G$. For a cut $V_1$, $V_2$, the set of edges that have exactly one end-vertex
in $V_1$ is denoted by $\cut(V_1)$, and the edges of $\cut(V_1)$ are called
\emph{cut-edges}.
When there is no danger of confusion, we also say that a single subset $V_1 \subseteq V$ is a cut
(meaning the partition $V_1$, $V\setminus V_1$).

\subsection{Easy cases}
\label{subs:easy}

In this subsection we present several results about easy special cases of the
problems that we focus on. This complements our hardness results.

The following theorem was proved by Ehrenfeucht et~al.~\cite{EHHR00-issues} and
also independently (in a slightly weaker form) by Kratochv\'\i l~\cite{Kratochvil03}.

\begin{theorem}\label{thm:ehhr-reg}
Let $\Cal P$ be a graph property that can be decided in time $\Cal O(n^a)$ for an
integer $a$. Let every graph with $\Cal P$ contain a vertex of degree at most
$d(n)$. Then the problem if an input graph is switching-equivalent to a graph with
$\Cal P$ can be decided in time $\Cal O(n^{d(n) + 1 + \max(a,2)})$.
\end{theorem}

The proof of Theorem~\ref{thm:ehhr-reg} also gives an algorithm that works in
the given time. Hence, it also provides an algorithm for
\textsc{Switch-Few-Edges}: in a graph with at most $k$ edges all vertex degrees
are bounded by $k$. Hence, we can use $d(n) = k$ and $a=2$ and get an $\Cal
O(n^{k+3})$-time algorithm. It was further proved by Jel\'\i nkov\' a et al.~\cite{JSHK11-param} that
\textsc{Switch-Few-Edges} is fixed-parameter tract\-able; it has a kernel with
$2k$ vertices, and there is an algorithm running in time $\Cal O(2.148^k \cdot n +m)$, where $m$ is the number
of edges of the input graph. In Section~\ref{sec:swfew-npc}, we provide a corrected
NP-completeness proof.

The following proposition states a basic relation of
switching-minimality and graph degrees.

\begin{proposition}[folklore]\label{pro:pul}
 Every switching-minimal graph $G=(V,E)$ on $n$ vertices has maximum degree at
most~$\lfloor (n-1)/2 \rfloor$.
\end{proposition}
\begin{proof}
 Clearly, if $G$ contains a vertex $v$ of degree greater than $\lfloor (n-1)/2
\rfloor$, then $S(G,\{v\})$ has fewer edges than $G$, showing that $G$ is not
switching-minimal.
\maybeqed\end{proof}

We remark that for a given graph $G$ we can efficiently construct a switch
whose maximum degree is at most $\lfloor (n-1)/2 \rfloor$; one by one, we switch
vertices whose degree exceeds this bound (in this way, the number of edges is
decreased in each step). However, the graph
constructed by this procedure is not necessarily switching-minimal.

Let $e(A)$ denote the number of edges whose one vertex is in $A$ and the
other one in $V(G) \setminus A$. The next proposition is an equivalent
formulation of Lemma 2.5 of Kozerenko~\cite{Kozerenko15}, strengthening
Proposition~\ref{pro:pul}.

\begin{proposition}\label{prop:kozerenko}
A graph $G$ is switching-minimal if and only if for every $A \subseteq V(G)$, 
 we have
\begin{displaymath}
2e(A) \leq |A|(|V(G)| -  |A|).
\end{displaymath}
\end{proposition}

We derive the following consequence.

\begin{proposition}\label{prop:maxdeg}
Let $G$ be a graph with $n$ vertices. If the maximum vertex degree in $G$ is at most
$\frac n 4$, then $G$ is switching-minimal.
\end{proposition}
\begin{proof}
Let $A$ be any subset of $V(G)$. We observe that $e(A) = e(V(G) \setminus A)$;
hence we can assume without loss of generality that $|A| \leq n/2$, and thus
$|V(G)| -  |A| \geq n/2$.

Further, as $e(A) \leq \sum_{v\in A}\mathop{{\rm deg}}(v)$, we have that $e(A) \leq |A|{\frac
n 4}$. Hence, $2e(A) \leq |A|(|V(G)| -  |A|)$, and the condition of
Proposition~\ref{prop:kozerenko} is fulfilled.
\maybeqed\end{proof}

Proposition~\ref{prop:maxdeg} implies that \textsc{Switch-Few-Edges} and
\textsc{Switch-Minimal} are trivially solvable in polynomial time for graphs on
$n$ vertices with maximum degree at most~$\frac n 4$.

We note that in Proposition~\ref{prop:maxdeg}, the bound $\frac n 4$ in general
cannot be improved, as shown by the example of a $k$-regular bipartite graph on
$n$ vertices with $k>\frac n 4$. Such a graph is switching-equivalent to a
$(\frac n 2 -k)$-regular bipartite graph, and therefore is not
switching-minimal.


\section{NP-Completeness of {\sc Switch-Few-Edges}}
\label{sec:swfew-npc}

Jel{\'\i}nkov{\' a} et al.~\cite{JSHK11-param} presented a proof that the problem
\textsc{Switch-Few-Edges} is NP-complete. Unfortunately, there is an error in their proof.
We present another proof here. The core of the original proof is a reduction from the \textsc{Max-Cut} problem.
Our reduction works in a similar way. However, we need the following more special version of
\textsc{Max-Cut} (we prove the NP-completeness of {\sc Large-Deg-Max-Cut} in
Section~\ref{sec:npc-large-max-cut}).

\vskip\fboxsep
\vskip\fboxsep
\noindent
\framebox[\hsize][l]{\begin{minipage}{\sirka}
\textsc{Large-Deg-Max-Cut}\\
 \textbf{Input:} A graph $G$ with $2n$ vertices such that the minimum vertex
degree of $G$ is $2n-4$ and the complement of $G$ does not contain triangles; an integer $j$\\
 \textbf{Question:} Does there exist a cut $V_1$ of $V(G)$ with at least $j$ cut-edges?
 \end{minipage}}
\vskip\fboxsep
\vskip\fboxsep

\begin{proposition}
\label{prop:eq}
Let $G$ be a graph. In polynomial time, we
can find a graph $G'$ such that $|V(G')| = 4|V(G)|$ and the following
statements are equivalent for every integer $j$:
\begin{enumerate}
\item[(a)] There is a cut in $G$ with at least $j$ cut-edges,
\item[(b)] there exists a set $A\subseteq V(G')$ such that $S(G',A)$
contains at most $|E(G')| - 16j$ edges.
\end{enumerate}
\end{proposition}

\begin{proof}
We first describe the construction of the graph~$G'$. For each vertex
$u$ of $G$ we create a corresponding four-tuple $\{u', u'', u''', u''''\}$ of
pairwise non-adjacent vertices in~$G'$. An edge of $G$ is then represented by a
complete bipartite graph interconnecting the two four-tuples, and a non-edge in
$G$ is represented by $8$ edges that form a cycle that alternates between the
two four-tuples (see Fig.~\ref{figure:udelatka}).

We remark that our construction of $G'$ follows a similar idea as the
construction in the attempted proof of Jel{\'\i}nkov{\' a} et al.~\cite{JSHK11-param},
a notable difference being that in the original construction, a vertex of $G$
was replaced by a pair of vertices of $G'$ rather than a four-tuple.


\begin{figure}
\centering
\includegraphics{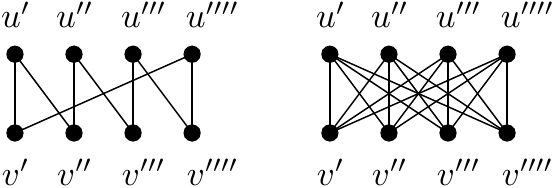}
\caption{The representation of non-edges and edges of $G$.\label{figure:udelatka}}
\end{figure}

A vertex four-tuple in $G'$ corresponding to a vertex of $G$ is called an
\emph{o-vertex}. A pair of o-vertices
corresponding to an edge of $G$ is called an \emph{o-edge} and a pair of o-vertices
corresponding to an non-edge of $G$ is called an \emph{o-non-edge}. Where there is no danger of confusion, we
identify o-vertices with vertices of $G$, o-edges with edges of $G$ and o-non-edges with non-edges of $G$.

We now prove that the statements (a) and (b) are equivalent.
First assume that there is a cut $V_1$ of $V(G)$ with $j'$ cut-edges. Let $V_1'$
be the set of vertices $u', u'', u''', u''''$ for all $u \in V_1$.
We prove that $S(G',V_1')$ contains at most $|E(G')| - 16j'$ edges.

We say that a non-edge \emph{crosses the cut $V_1$} if the non-edge has exactly one
vertex in $V_1$.
It is clear that $G'$ contains 16 edges per every o-edge and 8 edges per every o-non-edge.
In $S(G',V_1')$, every o-edge corresponding to an edge that is not a cut-edge is unchanged
by the switch and yields 16 edges. Similarly, every o-non-edge corresponding to a non-edge that
does not cross the cut yields 8 edges.

Fig.~\ref{figure:ctv-legalni} illustrates the switches of o-non-edges and
o-edges that have exactly one end-o-vertex in $V_1$.
We can see that every o-non-edge corresponding
to a non-edge that crosses the cut yields 8 edges in $S(G',V_1')$, and that every o-edge corresponding
to a cut-edge yields 0 edges. Altogether, $S(G',V_1')$ has $|E(G')| - 16j'$ edges,
which we wanted to prove.

\begin{figure}
\centering
\includegraphics{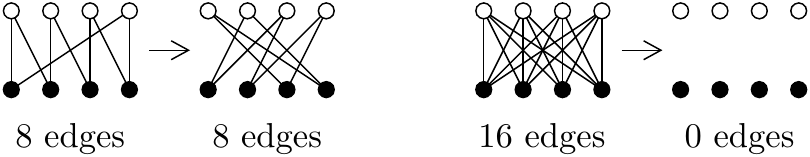}
\caption{Switches of an o-non-edge and of an o-edge.\label{figure:ctv-legalni}}
\end{figure}

Now assume that there exists a set $A\subseteq V(G')$ such that $S(G',A)$
contains at most $|E(G')| - 16j$ edges. We want to find a cut in $G$ with at
least $j$ cut-edges.

We say that an o-vertex $u$ of $G'$ is \emph{broken in $A$} if $A$ contains
exactly one, two or three vertices out of $u', u'', u''', u''''$; otherwise, we say that $u$ is \emph{legal in $A$}.
We say that an o-edge or o-non-edge $\{u,v\}$ is \emph{broken in $A$} if at
least one of the o-vertices $u$, $v$ is broken. Otherwise, we say that $\{u,v\}$ is \emph{legal in $A$}.

If all vertices of $G$ are legal in $A$, we say that $A$ is \emph{legal}.
Legality is a desired property, because for a legal set $A$ we can define a subset
$V_A$ of $V(G)$ such that 
\[
V_A =\left\{u \in V(G): \{u',u'',u''',u''''\}\subseteq A\right\}. 
\]
The set $V_A$
then defines a cut in $G$.
If a set is not legal, we proceed more carefully to get a cut from it.
For any vertex subset $A$, we say that a set $A'$ is a \emph{legalization} of $A$ if $A'$
is legal and if $A'$ and $A$ differ only on o-vertices that are broken in $A$.

We want to show that for every illegal set $A$, there exists its legalization
$A'$ such that the number of edges in $S(G',A')$ is not much higher than in
$S(G',A)$. To this end, we give the Algorithm Legalize
which for a set $A$ finds such a legalization $A'$. 
During the run of the Algorithm, we keep a set $A''$. In the beginning we set
$A'' := A$ and in each step we change $A''$ so that more o-vertices are legal.

We define some notions needed in the Algorithm. Let $v$ be an o-vertex and consider the o-vertices that are adjacent to $v$
(through an o-edge); we call them \emph{o-neighbors} of $v$. The o-neighbors of
$v$ are four-tuples of vertices and some of those vertices are in $A''$, some of them are not.
We define $\dif(v)$ as the number of such vertices that are in $A''$ minus the number
of such vertices that are not
in $A''$. (Note that $\dif(v)$ is always an even number, because the total number of vertices
in o-neighbors is even. If all o-neighbors were legal, then $\dif(v)$ would be divisible by four.)

\begin{figure}
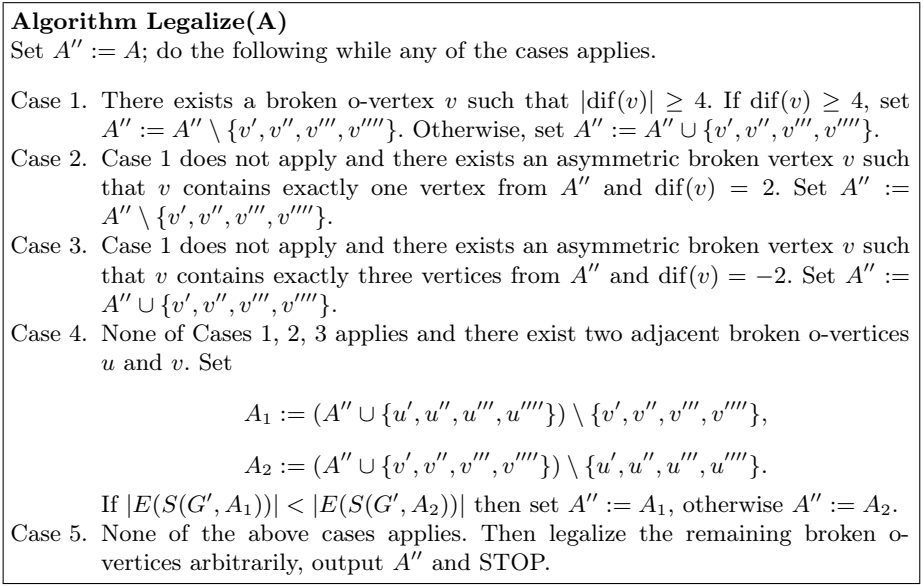

\centering
\framebox[\hsize][l]{\begin{minipage}{\sirka}
\textbf{Algorithm Legalize(A)}\\
Set $A'' := A$; do the following while any of the cases applies.
\begin{enumerate}[label=Case \arabic*.,leftmargin=*]
\item\label{case:a} 
There exists a broken o-vertex $v$ such that $|\dif(v)| \geq 4$. If $\dif(v) \geq 4$,
set $A'' := A'' \setminus \{v', v'', v''', v''''\}$. Otherwise,
set $A'' := A'' \cup \{v', v'', v''', v''''\}$.
\item\label{case:b} 
Case 1 does not apply and there exists an asymmetric broken vertex $v$ such that
$v$ contains exactly one vertex from $A''$ and $\dif(v) = 2$. Set $A'' := A''
\setminus \{v', v'', v''', v''''\}$.
\item\label{case:c} 
Case 1 does not apply and there exists an asymmetric broken vertex $v$ such that
$v$ contains exactly three vertices from $A''$ and $\dif(v) = -2$. Set $A'' := A''
\cup \{v', v'', v''', v''''\}$.
\item\label{case:d} 
None of Cases 1, 2, 3 applies and there exist two adjacent broken
o-vertices $u$ and $v$. Set
$$A_1 := (A'' \cup \{u', u'', u''', u''''\}) \setminus \{v', v'', v''', v''''\},$$
$$A_2 := (A'' \cup \{v', v'', v''', v''''\}) \setminus \{u', u'', u''', u''''\}.$$
If $|E(S(G',A_1))| < |E(S(G',A_2))|$ then set $A'' := A_1$, otherwise $A'' := A_2$.
\item\label{case:e} 
None of the above cases applies. Then legalize the remaining broken
o-vertices arbitrarily, output $A''$ and STOP.
\end{enumerate}
 \end{minipage}}
\caption{The Algorithm Legalize.\label{figure:alg}}
\end{figure}

The Algorithm is given in Fig.~\ref{figure:alg}.
As in the last step the Algorithm legalizes all remaining broken o-vertices, it
is clear that the set $A''$ output by the Algorithm is a
legalization of $A$.
We prove that $|E(S(G',A''))| - |E(S(G',A))| \leq 7$.

We need to introduce more terminology. A pair of vertices of $G'$ which belong to the same o-vertex is called a
\emph{v-pair}. A pair of vertices of $G'$ which belong to different
o-vertices that are adjacent (in $G$) is called an \emph{e-pair}.
A pair of vertices of $G'$ which belong to different o-vertices that are
non-adjacent (in $G$) is called an \emph{n-pair}. It is easy to see that any
edge of $G'$ or $S(G',A'')$ is either a v-pair, an e-pair or an n-pair. We call such edges \emph{v-edges},
\emph{e-edges} and \emph{n-edges}, respectively.

We say that a broken o-vertex $v$ is \emph{asymmetric} if it contains an odd number of
vertices of $A''$; we say that a broken o-vertex
is \emph{symmetric} if it contains two vertices out of $A''$.

To measure how the number of edges of $S(G',A'')$ changes during the run of the
Algorithm, we define a variable $\ch(A'')$ which we call the \emph{charge} of
the graph $S(G',A'')$.
Before the first step we set $\ch(A'') := |E(S(G',A))|$. After a step of the
Algorithm, we update $\ch(A'')$ in the following way.

\begin{itemize}
\item For every v-pair or e-pair that was an edge of $S(G',A'')$ before the step
and is no longer an edge of $S(G',A'')$ after the step, we decrease $\ch(A'')$
by one.
\item For every v-pair or e-pair that was not an edge of $S(G',A'')$ before the step
and that has become an edge of $S(G',A'')$ after the step, we increase $\ch(A'')$
by one.
\item For every o-vertex $v$ that was legalized in the step and is incident to an o-non-edge,
we change $\ch(A'')$ in the following way:
\begin{itemize}
\item if $v$ is symmetric, we increase $\ch(A'')$ by $2.5$ for every o-non-edge
incident to $v$;
\item if $v$ is asymmetric, we increase $\ch(A'')$ by $1.5$ for every o-non-edge
incident to $v$.
\end{itemize}
\end{itemize}

\begin{figure}
\centering
\includegraphics{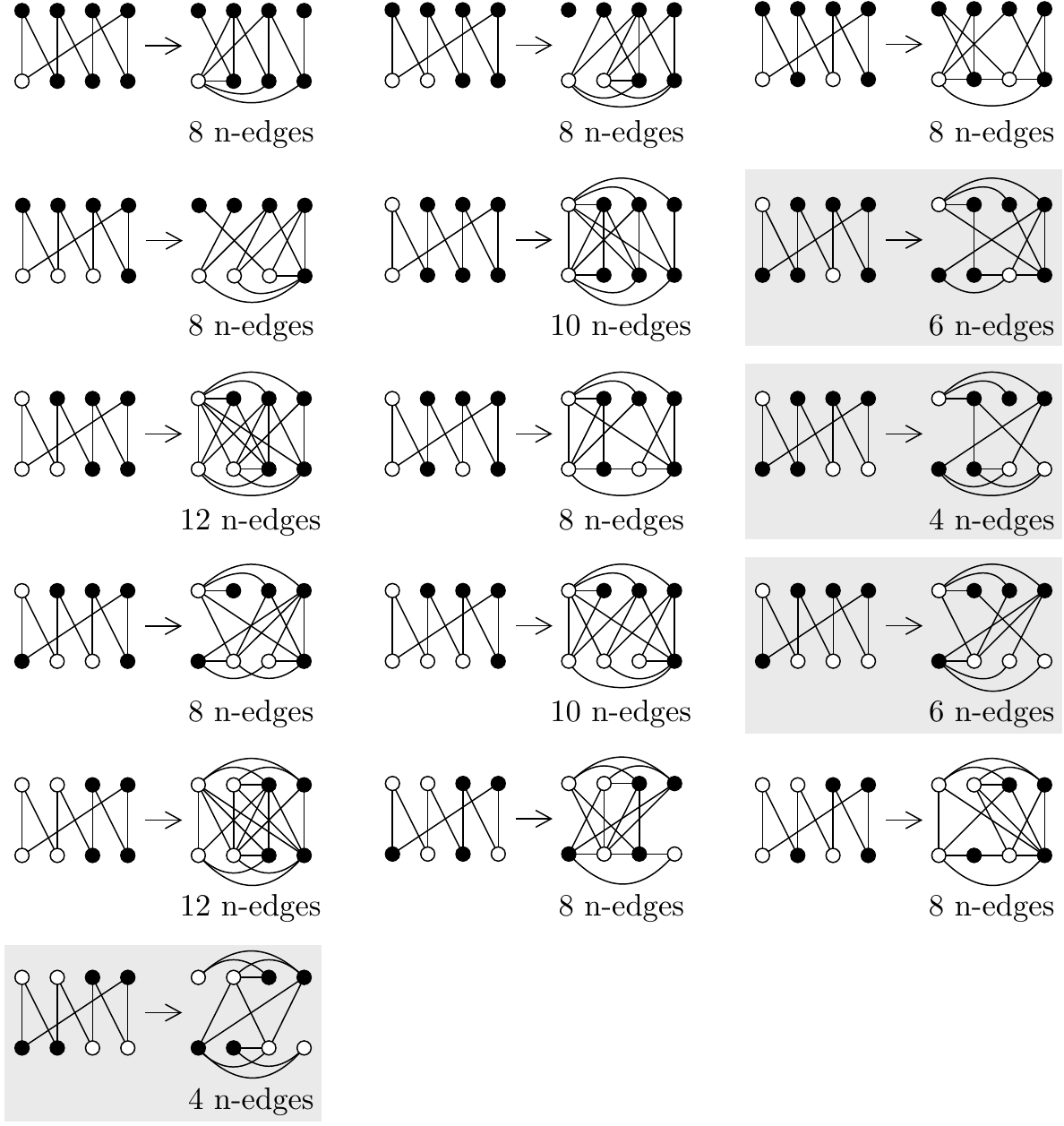}
\caption{All possible illegal switches of o-non-edges (up to symmetry). Vertices
of $A$ are marked in white and edges are as in $G$ (left to the arrow) and as in
$S(G',A)$ (right to the arrow). In the highlighted cases, the number of n-edges in $S(G',A)$ is lower
than~$8$.\label{figure:ctv-non-edges}}
\end{figure}

To explain the last two points, we observe how the number of n-edges increases
after legalizing an o-vertex. By analyzing all cases of o-non-edges with one or
two broken end-o-vertices (see Fig.~\ref{figure:ctv-non-edges}),
we get that there are four cases where the o-non-edges have less than 8
n-edges before legalization: either 6 or 4 n-edges. In these cases, both end-o-vertices are broken.
If there are only 4 n-edges, at least one of the end-o-vertices is symmetric.
After one end-o-vertex is legalized,
the number of n-edges increases by 2 or 4. When the second end-o-vertex is legalized, the
number of n-edges does not increase for this particular o-non-edge.

After both end-o-vertices are legalized, the charge has been changed in the
following way: if both end-o-vertices were symmetric, we have
increased the charge by $5$. If one of them was symmetric and the other one was
asymmetric, we have increased the charge by $4$. Finally, if both were
asymmetric, we have increased the charge by $3$. In all these cases, 
the increase is an upper bound on the number of contributed n-edges.

Further, every v-edge or e-edge that has appeared or disappeared during the run of the Algorithm
is counted immediately after the corresponding step. Hence, we have proved the following Claim.

\begin{claimm}\label{cl:upper-bound}
At the end of the Algorithm we have that $\ch(A'') \geq |E(S(G',A''))|$.
\end{claimm}

Next, we give an upper bound on the charge $\ch(A'')$.

\begin{claimm}\label{cl:decrease}
After every step of the Algorithm except for the last one, the charge $\ch(A'')$ is decreased.
After the last step, the charge is increased by at most 7. Hence, $\ch(A'') \leq |E(S(G',A))| + 7$.
\end{claimm}

To prove Claim~\ref{cl:decrease}, we count how the charge changes after each
step of the Algorithm Legalize. We distinguish cases according to which the step was done.

\begin{enumerate}[label=Case \arabic*.,leftmargin=*]
\item 
	We may assume without loss of generality that $\dif(v)\geq 4$ (otherwise we swap the roles of $A''$ and $V(G') \setminus A''$.
	Further, $v$ can be either symmetric or asymmetric; we first assume that $v$ is symmetric (see Fig.~\ref{figure:case1}). Then by its legalization the number
	of v-edges is decreased by $4$.

	As $\dif(v)\geq 4$, then among vertices in o-neighbors of $v$, there must be at least four more vertices
	belonging to $A''$ than those not belonging to $A''$. Thus, by removing any vertex of $\{v', v'', v''', v''''\}$ from $A''$
	we reduce the number of e-edges by at least $4$. As $v$ contains two vertices out of $\{v', v'', v''', v''''\}$, we reduce the number
	of e-edges by at least $8$.

	For n-pairs that have one vertex inside $v$ the charge is increased by at most $3 \cdot 2.5$, which is $7.5$.
	To sum it up:
	\begin{itemize}
	\item For v-pairs the charge is decreased by $4$,
	\item for e-pairs the charge is decreased by at least $8$,
	\item for n-pairs the the charge is increased by at most $7.5$.
	\end{itemize}
	Altogether, the charge is decreased by at least $4.5$.

	If the o-vertex $v$ is asymmetric, then in an analogical way we have
	that for v-pairs the charge is decreased by 3, for e-pairs the charge is decreased by at least 4,
	and for n-pairs the the charge is increased by at most $4.5$.
	Altogether, the charge is decreased by at least $2.5$.

\item 
	The analysis is similar as above. We get $-3$ for v-pairs, $-2$ for e-pairs, and $\leq 4.5$ for n-pairs.
	Altogether, the charge is decreased by at least $0.5$.
\item 
This case is symmetric to Case 2. Hence, the charge is decreased by at least $0.5$ as well.
\item 
	In this case, when counting how the charge was changed because of e-pairs, we need to bound both
	the number of e-edges between a vertex in $u$ and a vertex in $v$, and the number of e-edges
	between a vertex inside $u$ or $v$ and a vertex inside one of their other o-neighbors.
	This depends also on the values of $\dif(u)$ and $\dif(v)$.

	We analyze four subcases of o-edges whose both end-o-vertices are
	broken -- they are numbered as in Fig.~\ref{figure:ctv-edges}.

	\begin{enumerate}[label=\Roman*.,leftmargin=*]
	\item 

\begin{figure}[t]
\centering
\includegraphics{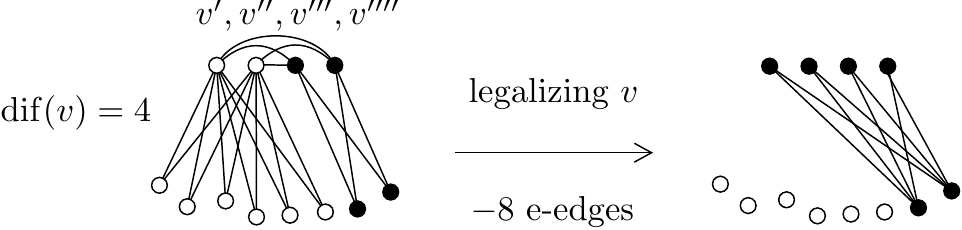}
\caption{A simplified illustration to the analysis of Case 1. Vertices of $A''$ are marked in white, and edges are as in $S(G',A'')$
before the step (left to the arrow) and after the step (right to the arrow).
\label{figure:case1}}
\end{figure}

\begin{figure}[t]
\centering
\includegraphics{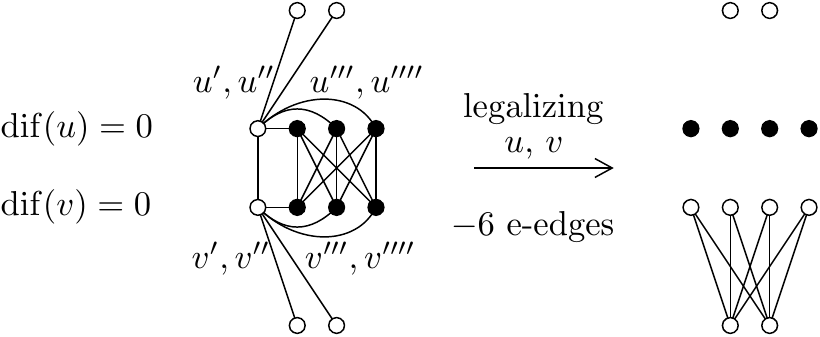}
\caption{A simplified illustration to the analysis of Case 4, I. Vertices of $A''$ are marked in white, and edges are as in $S(G',A'')$ before the step (left to the
arrow) and after the step (right to the arrow).\label{figure:caseI}}
\end{figure}

\begin{figure}
\centering
\includegraphics{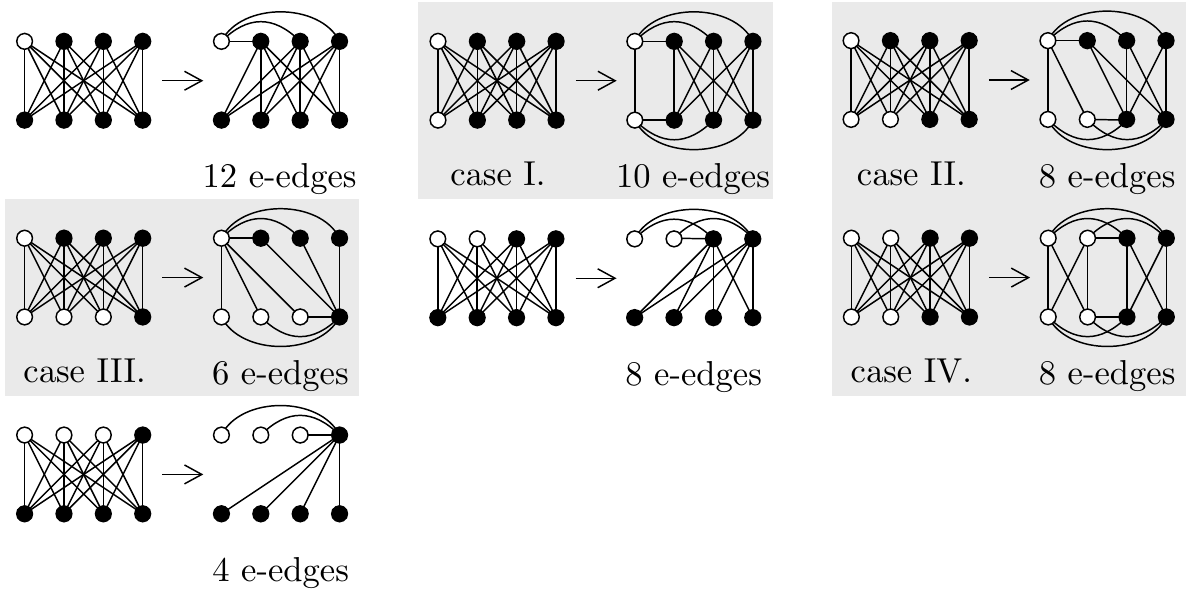}
\caption{All possible illegal switches of o-edges (up to symmetry). Vertices
of $A$ are marked in white and edges are as in $G$ (left to the arrow) and as in
$S(G',A)$ (right to the arrow). In the highlighted cases, both end-o-vertices
are broken.\label{figure:ctv-edges}}
\end{figure}

	First assume that $\dif(u) = 0$ and $\dif(v) = 0$ (see Fig.~\ref{figure:caseI}). We can see that vertices
	inside $v$ contribute by $-2$ to $\dif(u)$. Hence, outside $v$, there must be two more
	vertices in o-neighbors of $u$ that are in $A''$ than those not in $A''$.
	The same holds symmetrically for o-neighbors of $v$ outside $u$.

	We may without loss of generality assume that the Algorithm chose to set
	$$A'' := (A'' \cup \{v', v'', v''', v''''\}) \setminus \{u', u'', u''', u''''\}.$$
	Then, the number of e-pairs adjacent to both $u$ and $v$ is decreased by 10; the number of e-pairs adjacent to $u$ and not to $v$
	is decreased by $2$, and the number of e-pairs adjacent to $v$ and not to $u$ is increased by $6$. Altogether,
	the charge is decreased by $6$ for e-pairs.

	For v-pairs, the charge is decreased by $6$, and for n-pairs, the charge is increased by at most $6\cdot 1.5$. Altogether,
	the charge is decreased by at least $3$.

	It remains to analyze the cases when $\dif(u)$ and $\dif(v)$ are different. As neither Case 2 nor Case 3 applies, we know
	that none of $\dif(u)$, $\dif(v)$ is equal to $2$.

	By analogical ideas as above, we get that if one of $\dif(u)$, $\dif(v)$ is equal to $0$ and the other one to $-2$,
	the charge is decreased by at least $9$. If both $\dif(u)$, $\dif(v)$ are equal to $-2$, then the charge is decreased by at least $7$.
	\item[II.] 
	As $u$ is asymmetric and $v$ is symmetric, we have that for n-pairs the
	charge is increased by $3 \cdot 2.5 + 3 \cdot 1.5$, which is $12$. For v-pairs, the
	charge is decreased by $3 + 4$, which is $7$.

	We consider the case when the Algorithm chose to set
	$$A'' := A'' \cup \{u', u'', u''', u''''\} \setminus \{v', v'', v''', v''''\}$$
	(if we get a sufficient bound for this case, then the other case could
	only be better).

	To count the decrease for e-pairs, we need to consider the values of
	$\dif(u)$ and $\dif(v)$.
	Assume that $\dif(u) = 0$. Then, outside $v$, there must be the same number of vertices in
	o-neighbors of $u$ that are in $A''$ as those that are not in $A''$.
	
	If $\dif(v) = 0$, then outside $u$ there must be two more vertices in o-neighbors of $v$ that are in
	$A''$ than those not in $A''$. Then for e-pairs, the charge is decreased
	by $12$. If $\dif(v) = 2$, then using analogous ideas we get that for e-pairs, the charge is decreased
        by $16$. If $\dif(v) = -2$, we get $8$.

	Now assume that $\dif(u) = -2$. By considering the number of vertices in
	o-neighbors of $u$ and $v$, we get that the charge decrease for e-pairs is
	either $14$ (if $\dif(v) = -2$) or $18$ (if $\dif(v) = 0$) or $22$ (if
	$\dif(v)= 2$).

	As Case 2 does not apply, we know that $\dif(u)$ is not equal to $2$.
	Hence, we have considered all the cases, and the charge decrease for
	e-pairs is at least $8$. Altogether, the charge is decreased by at least
	$-12 +7 + 8$, which is $3$.

	\item[III.] 
	As both $u$ and $v$ are asymmetric, we have that for n-pairs the charge is increased
	by $6 \cdot 1.5$, which is $9$. For v-pairs, the charge is decreased by
	$3 + 3$, which is $6$.

	Again, we consider the case when the Algorithm chose to set
        $$A'' := A'' \cup \{u', u'', u''', u''''\} \setminus \{v', v'', v''', v''''\}.$$

	By using the same idea as above, we get that for e-pairs, the charge is
	decreased by $18$ (if $\dif(u) = 0$ and $\dif(v) = 0$), or by $24$ (if
	$\dif(u) = -2$ and $\dif(v) = 0$, or if $\dif(u) = 0$ and $\dif(v) =
	2$), or by $30$ (if $\dif(u) = -2$ and $\dif(v) = 2$).

	As Case 2 does not apply, we know that $\dif(u)$ cannot be $2$ and
	$\dif(v)$ cannot be $-2$. Hence, we have considered all the cases and for e-pairs, the charge is decreased by
	at least $18$. Altogether, the charge is decreased by at least
	$-9 + 6 + 18$, which is $15$.

	\item[IV.] 
	As both $u$ and $v$ are symmetric, we have that for n-pairs the charge is increased
	by $6 \cdot 2.5$, which is $15$. For v-pairs, the charge is decreased by
	$4 + 4$, which is $8$.

	Without loss of generality, we consider only cases when $\dif(u) \leq \dif(v)$ (the other cases
	are symmetric). Thus, we may limit ourselves again to the case when the Algorithm
	chose to set
        $$A'' := A'' \cup \{u', u'', u''', u''''\} \setminus \{v', v'', v''', v''''\}.$$
	
	If $\dif(u) = \dif(v)$, then we easily check that the charge decrease for e-pairs
	is $8$. If $\dif(u) = 0$ and $\dif(v) = 2$ then the charge decrease for
	e-pairs is $12$. If $\dif(u) = -2$ and $\dif(v) = 0$ then the decrease is $12$,
	and if $\dif(u) = -2$ and $\dif(v) = 2$ then the decrease is $16$.

	Altogether, the charge decrease for e-pairs is at least $8$, and the
	total decrease is at least $-15 + 8 + 8$, which is $1$. 
	\end{enumerate}

\item 
If Case 5 applies, then all remaining broken o-vertices must be pairwise non-adjacent
(because Case~4 does not apply). Hence, there must be at most two broken o-vertices left
(otherwise, there would be a triangle in the complement of the input graph,
which would contradict the assumptions).
Further, each of these o-vertices has $\dif = 0$, because all its o-neighbors are
legal and Case 1 does not apply. Thus, the charge change for e-pairs due to this
last step is $0$.

To count the charge change for n-pairs and v-pairs, we analyze the five cases (one or two o-vertices, symmetric or asymmetric). 
If there is one symmetric o-vertex left, then the charge increase for n-pairs is
$3 \cdot 2.5$ and the decrease for v-pairs is $4$, hence the total increase is $7.5
- 4$, which is $3.5$. If there is one asymmetric o-vertex, then the total
increase is $3\cdot 1.5 - 3$, which is $1.5$.

If there are two broken o-vertices left and both are asymmetric, then
the total increase is $6\cdot 1.5 - 6$, which is $3$. If one of them is
symmetric and the other one is asymmetric we get $5$; if both are symmetric, we
get~$7$. Altogether, we get that the charge is increased by at most 7.
\end{enumerate}

We have proved Claim~\ref{cl:decrease}. Further, by Claim~\ref{cl:upper-bound} and
Claim~\ref{cl:decrease} we have that $|E(S(G',A''))| \leq |E(S(G',A))| + 7$, and hence $A''$ is the sought legalization of
$A$.

We continue the proof of Proposition~\ref{prop:eq}.
We have already argued that a legal set $A''$ defines a subset $V_{A''}$ of $V(G)$, and hence a cut in $G$.
Assume that $\cut(V_{A''})$ has $j'$ edges.
From the proof of the first implication of Proposition~\ref{prop:eq} we know that the number of edges in
$S(G',A'')$ can be expressed as $|E(G')| - 16j'$.

On the other hand, we have proved that the number of edges in $S(G',A'')$ is at most $|E(G')| - 16j + 7$. We get that
$|E(G')| - 16j'  \leq  |E(G')| - 16j + 7$, and hence $j' \geq j - 7/16$. As both $j$ and $j'$ are
integers, we have that $j' \geq j$. Hence, $\cut(V_{A''})$ has at least $j$
edges, and Proposition~\ref{prop:eq} is proved.
\maybeqed\end{proof} 

\begin{theorem}
\label{thm:sw-few-npc}
{\sc Switch-Few-Edges} is NP-complete.
\end{theorem}
\begin{proof}
Theorem~\ref{thm:large-deg-max-cut} in the next section gives the NP-completeness of
\textsc{Large-Deg-Max-Cut}. Further, by Proposition~\ref{prop:eq}, an instance
$(G,j)$ of \textsc{Large-Deg-Max-Cut} can be transformed into an instance
$(G',j')$ of {\sc Switch-Few-Edges} such that there is a cut in $G$ with at
least $j$ cut-edges if and only if $G'$ is $(\leq j')$-switchable. The
transformation works in polynomial time.

Finally, it is clear that the problem {\sc Switch-Few-Edges} is in NP.
\maybeqed\end{proof}

\section{The NP-Completeness of {\sc Large-Deg-Max-Cut}}\label{sec:npc-large-max-cut}


Let $G$ be a graph with $2n$ vertices. A \emph{bisection of $G$} is a partition
$S_1$, $S_2$ of $V(G)$ such that $|S_1| = |S_2| = n$ (hence, a bisection is a
special case of a cut). The size of $\cut(S_1)$
is called the \emph{size} of the bisection $S_1$, $S_2$. A \emph{minimum bisection
of $G$} is a bisection of $G$ with minimum size.

Garey et al.~\cite{GJS76} proved that, given a graph $G$ and an integer $b$, the
problem to decide if $G$ has a bisection of size at most $b$ is NP-complete (by a reduction
of \textsc{Max-Cut}). Their formulation is slightly different from ours -- two distinguished
vertices must be each in one part of the partition, and the input graph does not
have to be connected. However, their reduction from \textsc{Max-Cut}
(see~\cite[pages 242--243]{GJS76}) produces only
connected graphs as instances of the bisection problem, and it is immediate that
the two distinguished vertices are not important in the proof. Hence, their
proof gives also the NP-completeness of the following version of the problem.

\vskip\fboxsep
\vskip\fboxsep
\noindent
\framebox[\hsize][l]{\begin{minipage}{\sirka}
\textsc{Connected-Min-Bisection}\\
 \textbf{Input:} A connected graph $G$ with $2n$ vertices, an integer $b$\\
 \textbf{Question:} Is there a bisection $S_1$, $S_2$ of $V(G)$ such that
$\cut(S_1)$ contains at most $b$ edges?
 \end{minipage}}
\vskip\fboxsep
\vskip\fboxsep

From the NP-completeness of \textsc{Min-Bisection}, Bui et al.~\cite{BCLS87}
proved the NP-completeness of \textsc{Min-Bisection} restricted to 3-regular
graphs (as a part of a more general result, see~\cite[proof of Theorem
2]{BCLS87}). We use their result to prove the NP-completeness of
\textsc{Large-Deg-Max-Cut}.


\vskip\fboxsep
\vskip\fboxsep
\noindent
\framebox[\hsize][l]{\begin{minipage}{\sirka}
\textsc{Large-Deg-Max-Cut}\\
 \textbf{Input:} A graph $G$ with $2n$ vertices such that the minimum vertex
degree of $G$ is $2n-4$ and the complement of $G$ is connected and does not contain triangles; an integer $j$\\
 \textbf{Question:} Does there exist a cut $V_1$ of $G$ with at least $j$ cut-edges?
 \end{minipage}}
\vskip\fboxsep
\vskip\fboxsep

\begin{lemma}\label{lemma:bisection}
Let $G$ be a connected 3-regular graph on $2n$ vertices. Let $b$ be the size of the minimum bisection
in $G$ and let $c$ be the size of the maximum cut in~$\overline{G}$. Then $b =
n^2 - c$.
\end{lemma}
\begin{proof}
Let $S_1$, $S_2$ be a minimum bisection in $G$ and let $b$ be the size of the bisection.
In $\overline{G}$, the partition $S_1$, $S_2$ yields a cut with $n^2 - b$ cut-edges.
Hence, $c \geq n^2 - b$.

On the other hand, let $V_1$, $V_2$ be a maximum cut in $\overline{G}$ for which the sizes of $V_1$ and $V_2$ are
as close as possible. If $|V_1| = |V_2| = n$, the partition $|V_1|$, $|V_2|$ gives
a bisection in $G$ of size $n^2 - c$, hence $b \leq n^2 - c$ and we are done.

Otherwise, assume that $|V_1| = n - k$ and $|V_2| = n+k$ for a $k \geq 1$. As the graph $G$ is connected, there is a vertex
$v$ in $V_2$ that has at least one neighbor in $V_1$. We set $V_1' = V_1 \cup \{v\}$ and $V_2' = V_2 \setminus \{v\}$.

The vertex $v$ has at least one neighbor in $V_1$. Hence, in $G$, there is at least one edge between
$v$ and $V_1$, and in $\overline{G}$, there are at most $n-k-1$ cut-edges adjacent to $v$.
 
Further, $v$ has at most two neighbors in $V_2$ and at least $n+k-3$ non-neighbors in $V_2$.
Hence, in the partition  $V_1'$, $V_2'$, there will be at least $n + k - 3$ cut-edges adjacent to $v$.
Cut edges that are not adjacent to $v$ are the same in $V_1$, $V_2$ as in $V_1'$, $V_2'$.

Altogether, $|\cut(V_1')| - |\cut(V_1)| \geq n+ k - 3 - (n-k-1) \geq 2k-2 \geq 0$.
Hence, the partition $V_1'$, $V_2'$ has smaller difference of the sizes of the two parts while the size of the
cut is not smaller, which is a contradiction with the choice of $V_1$,
$V_2$.\qed
\end{proof}


\begin{theorem}\label{thm:large-deg-max-cut}
\textsc{Large-Deg-Max-Cut} is NP-complete.
\end{theorem}
\begin{proof}
Let $(G,b)$ be an instance of \textsc{Connected-Min-Bisection}. We use the construction
of Bui et al.~\cite[proof of Theorem 2]{BCLS87}. Their first step is to construct from
an instance $(G,b)$ of \textsc{Min-Bisection} a 3-regular graph
$G^*$ such that $G$ has a minimum bisection of size $b$ if and only if $G^*$ has
a minimum bisection of size $b$. Further, it is immediate from their construction that $G^*$
contains no triangles, and if $G$ is connected, then $G^*$ is connected as well.
Moreover, $G^*$ has an even number of vertices.

We see that $\overline{G^*}$ fulfills the conditions of an instance of
\textsc{Large-Deg-Max-Cut}. By Lemma~\ref{lemma:bisection} we know that $G^*$
has a minimum bisection of size $b$ if and only if $\overline{G^*}$ has a maximum
cut of size $m^2 - b$.

Altogether, $G$ has a minimum bisection of size $b$ if and only if
$\overline{G^*}$ has a maximum cut of size $m^2 - b$. Hence, $(\overline{G^*},
m^2 - b)$ is an equivalent instance of \textsc{Large-Deg-Max-Cut}.
To finish the proof that \textsc{Large-Deg-Max-Cut} is NP-complete, we observe that
\textsc{Large-Deg-Max-Cut} is in NP.
\maybeqed\end{proof}

\section{Switching of Graphs with Bounded Density}

The \emph{density} of a graph $G$ is defined as
\begin{displaymath}
D(G) = \frac{|E(G)|}{\binom{|V(G)|}{2}} =
\frac{2|E(G)|}{|V(G)|(|V(G)|-1)}.
\end{displaymath}

In connection with properties of simplicial complexes, Matou\v sek and Wagner~\cite{MW14}
asked if deciding switching-minimality was easy for graphs of bounded density. We give a partial negative answer
by proving that the problem \textsc{Switch-Few-Edges} stays NP-complete 
even for graphs of density bounded by an arbitrarily small constant. This is in
contrast with Proposition~\ref{prop:maxdeg}, which shows that any graph $G$ with
maximum degree at most $|V(G)|/4$ is switching-minimal. The core of our 
argument is the following Proposition.

\begin{proposition}\label{proposition:graphbc}
Let $G$ be a graph, let $k$ be an integer, and let $c$ be a fixed constant in
$(0,1)$. In polynomial time, we can find a graph $G'$ and an integer $k'$ such
that
\begin{enumerate}
\item\label{cond:dens} $D(G') \leq c$,
\item\label{cond:sw}   $G'$ is $(\leq k')$-switchable if and only if $G$ is $(\leq k)$-switchable,
\item\label{cond:min}   $G'$ is switching-minimal if and only if $G$ is
switching-minimal, and
\item\label{cond:size} $|V(G')| = O({|V(G)|})$.
\end{enumerate}
\end{proposition}

\begin{proof}
Let $n = |V(G)|$ and let
$N=\max\left\{n,\left\lceil\frac{3n}{4c}\right\rceil\right\}$. We construct the 
graph $G'$ in the following way (see also Fig.~\ref{figure:graphbc-construct}).
Let $V=V(G)$. Then 
\begin{displaymath}
V(G') = V \cup Y \cup Z,
\end{displaymath}
where $Y$ is a set of $N$ vertices and $Z$ is a set of $N$ more vertices, and
\begin{displaymath}
E(G') = \{ \{v_1,v_2\}\colon v_1 \in Y, v_2 \in V \} \cup E(G).
\end{displaymath}

\begin{figure}
\centering
\includegraphics{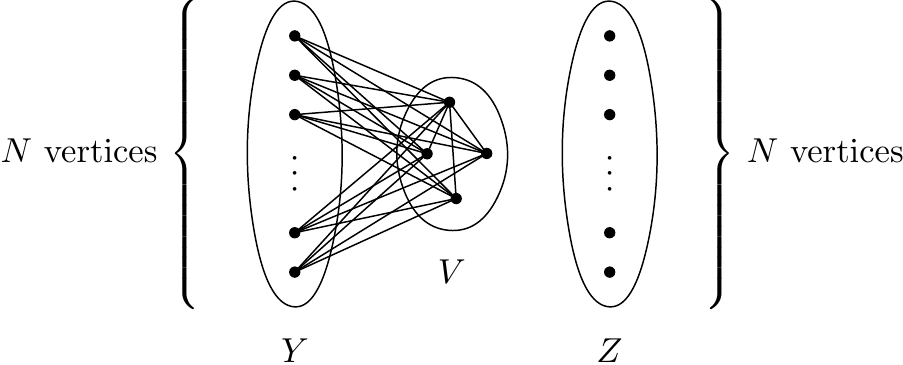}
\caption{The graph $G'$.\label{figure:graphbc-construct}}
\end{figure}

We prove that $G'$ fulfills the conditions of
Proposition~\ref{proposition:graphbc}. It is easy to see
that Condition~\ref{cond:size} holds and that $G'$ can be obtained in
polynomial time.
We prove that Conditions~\ref{cond:sw} and~\ref{cond:min} hold, too.

Assume that $G$ is switching-reducible, i.~e., there exists a set $A\subseteq V$
such that $S(G,A)$ contains fewer edges than~$G$. Let us count the number of
edges in $S(G',A)$.

It is easy to see that if we switch a subset of $V$ in $G'$, the number of
edges whose one endpoint is outside $V$ is unchanged, and the number of edges
with both endpoints outside $V$ remains zero.
We also observe that $S(G',A)[V]$ (the induced subgraph of $S(G',A)$ on the vertex subset $V$) is equal to $S(G,A)$. Hence, $S(G',A)$ has fewer
edges than $G'$, showing that $G'$ is switching-reducible.

Moreover, if $S(G,A)$ has $l$ edges for an integer $l$, then $S(G',A)$ has $l + nN$ edges. Thus, if
$G$ is $(\leq k)$-switchable, we have that $G'$ is $(\leq k+nN)$-switchable.

Now assume that $G'$ is switching-reducible, i.~e., there exists a set
$A\subseteq V(G')$ such that $S(G',A)$ has fewer edges than $G'$. If
$A\subseteq V$, we have that $S(G,A)$ has fewer edges than $G$, and Condition
\ref{cond:min} is satisfied. On the other hand, if $A \not\subseteq V$, we use 
the following Claim.

\begin{claimm}\label{claim:dobrehran}
Let $A$ be a subset of $V(G')$ and let $A' = A\cap V$. Then the number of edges in $S(G',A')$ is less than or equal to
the number of edges in $S(G',A)$.
\end{claimm}
To prove the claim, we fix a set $A\subseteq V(G')$. We may assume that $|A\cap
(Y\cup Z)| \le |Y\cup Z|/2=N$, otherwise we replace $A$ by its complement
$\overline{A}=V(G')\setminus A$ (note that $S(G',A\cap V)$ has the same number
of edges as $S(G',\overline{A}\cap V)$).

Define the sets $A'=A\cap V$ and $A''=A\setminus A' =A \cap(Y\cup Z)$. Let
$G'_1= S(G',A')$ and $G'_2=S(G',A)$. Note that $G'_2=S(G'_1,A'')$. To prove the
claim, we need to show that $G'_1$ has at most as many edges as $G'_2$.

In $G'_2$, every vertex of $A''$ is adjacent to every vertex of
$(Y\cup Z)\setminus A''$, whereas no such pair is adjacent in $G'_1$. This
means that $|E(G'_2)\setminus E(G'_1)|\ge |A''|(|Y|+|Z|-|A''|) \ge |A''|N$,
where we used the fact that $A''$ has size at most $N$.

On the other hand, an edge belonging to $G'_1$ but not to $G'_2$ must
necessarily connect a vertex from $A''$ with a vertex from $V$. Therefore,
$|E(G'_1)\setminus E(G'_2)|\le |A''|n$. Combining these estimates, we get
\begin{align*}
 |E(G'_2)| - |E(G'_1)| &= |E(G'_2)\setminus E(G'_1)| - |E(G'_1)\setminus
E(G'_2)|\\
&\ge |A''|N - |A''|n\\
&\ge 0.
\end{align*}
This proves the claim.
As a consequence of Claim~\ref{claim:dobrehran}, if $G'$ is switching-reducible,
then it can be reduced by switching a set $A'\subseteq V$. The same set $A'$
then reduces $G$, and Condition~\ref{cond:min} of the Proposition holds.
Analogically, if $G'$ can be switched to contain $L$ edges for an integer $L$,
then $G$ can be switched to contain $L - nN$ edges. Hence, we have proved Condition~\ref{cond:sw}
with $k' = k + nN$.

It remains to check Condition~\ref{cond:dens}. By definition, the density of
$G'$ is

\begin{align*}
D(G') &= \frac{ 2|E(G')|}{(2N+n)(2N+n-1)}\\ &\le \frac{2\left(\binom{n}{2} + nN\right)}{(2N+n)(2N+n-1)}\\
&\le \frac{n^2+ 2nN}{4N^2}\\ &\le \frac{3nN}{4N^2}= \frac{3n}{4N}\\ &\le c.
\end{align*}
This completes the proof.
%
%
%
%
%
\maybeqed\end{proof} 

Proposition~\ref{proposition:graphbc} allows us to state a stronger version of Theorem~\ref{thm:sw-few-npc}
for the special case of graphs with bounded density.

\begin{theorem}\label{theorem:npc-c}
For every $c>0$, the problem \textsc{Switch-Few-Edges} is NP-complete for graphs
of density at most $c$.
\end{theorem}
\begin{proof}
As shown by Proposition~\ref{proposition:graphbc}, a general instance $(G,k)$ of
\textsc{Switch-Few-Edges} can be transformed into an equivalent instance $(G',k')$ of
density at most~$c$. Since \textsc{Switch-Few-Edges} is NP-complete on general
instances by Theorem~\ref{thm:sw-few-npc}, it remains NP-complete on instances
of density at most~$c$.
\maybeqed\end{proof}


\section{Concluding Remarks}

{\bf 5.1.}
We have been trying to prove that the problem \textsc{Switch-Reducible} is NP-complete (and hence,
\textsc{Switch-Minimal} is co-NP-complete). We have not yet succeeded. However,
if it is true, then Proposition~\ref{proposition:graphbc} gives the following analogue of Theorem~\ref{theorem:npc-c}
even for these problems.

\begin{proposition}\label{theorem:swred-npc-c}
If the problem \textsc{Switch-Reducible} is NP-complete, then for every $c>0$, 
the problem \textsc{Switch-Reducible} is NP-complete for graphs
of density at most $c$, and the problem \textsc{Switch-Minimal} is
co-NP-complete for graphs of density at most $c$.
\end{proposition}

\medskip\noindent
{\bf 5.2.} Lindzey~\cite{Lindzey14} 
noticed that it is possible to speed-up several graph algorithms using
switching to a lower number of edges -- he obtained up to super-polylogarithmic
speed-ups of algorithms for diameter, transitive closure, bipartite maximum
matching and general maximum matching. However, he focuses on switching digraphs
(with a definition somewhat 
different to Seidel's switching in undirected graphs), 
 where the situation is in
sharp contrast with our results -- a digraph with the minimum number of edges in
its switching-class can be found in $O(n+m)$ time.

\medskip\noindent
{\bf 5.3.} It has been observed before (cf. e.g. \cite{EHHR00-issues}) that for a graph property $\cal P$, the complexity of deciding $\cal P$ is independent on the complexity of deciding if an input graph can be switched to a graph possessing the property $\cal P$. Switching to few edges thus adds another example of a polynomially decidable property (counting the edges is easy) whose switching version is hard. Previously known cases are the NP-hardness of  deciding switching-equivalence to a regular graph
\cite{Kratochvil03} and deciding switching-equivalence to an $H$-free graph for certain specific graphs $H$~\cite{JK14-hfree}.

\medskip\noindent
{\bf 5.4.}
Let $d>0$ be a constant. What can we say about the complexity of
\textsc{Switch-Reducible} and \textsc{Switch-Few-Edges} on graphs of maximum
degree at most~$dn$? If $d\le\frac 1 4$, the two problems are trivial by
Proposition~\ref{prop:maxdeg}. On the other hand, for $d\ge \frac 1 2$ the
restriction on maximum degree becomes irrelevant, since any switching-minimal
graph has maximum degree at most $\frac n 2$ by Proposition~\ref{pro:pul}. For
any $d\in(\frac 1 4,\frac 1 2)$, the complexity of the two problems on instances
of maximum degree at most $dn$ is open.



\bibliographystyle{plain}

\begin{thebibliography}{10}

\bibitem{BCLS87}
T.~N. Bui, S.~Chaudhuri, F.~T. Leighton, and M.~Sipser.
\newblock Graph bisection algorithms with good average case behavior.
\newblock {\em Combinatorica}, 7(2):171--191, 1987.

\bibitem{EHHR00-issues}
Andrzej Ehrenfeucht, Jurriaan Hage, Tero Harju, and Grzegorz Rozenberg.
\newblock {Complexity Issues in Switching of Graphs}.
\newblock In Hartmut Ehrig, Gregor Engels, Hans-J{\"o}rg Kreowski, and Grzegorz
  Rozenberg, editors, {\em {Theory and Application to Graph Transformations}},
  volume 1764 of {\em {LNCS}}, pages 59--70. Springer, Heidelberg, 2000.

\bibitem{GJS76}
M.~R. Garey, D.~S. Johnson, and L.~Stockmeyer.
\newblock Some simplified {NP}-complete graph problems.
\newblock {\em Theoretical Computer Science}, 1(3):237 -- 267, 1976.

\bibitem{Hage01}
Jurriaan Hage.
\newblock {\em {Structural Aspects Of Switching Classes}}.
\newblock PhD thesis, Leiden Institute of Advanced Computer Science, 2001.

\bibitem{JK14-hfree}
Eva Jel{\'i}nkov{\'a} and Jan Kratochv{\'i}l.
\newblock {On Switching to {H}-Free Graphs}.
\newblock {\em Journal of Graph Theory}, 75(4):387--405, 2014.

\bibitem{JSHK11-param}
Eva Jel{\'i}nkov{\'a}, Ond\v{r}ej Such{\'y}, Petr Hlin\v{e}n{\'y}, and Jan
  Kratochv{\'i}l.
\newblock {Parameterized Problems Related to {S}eidel's Switching}.
\newblock {\em Discrete Mathematics and Theoretical Computer Science},
  13(2):19--42, 2011.

\bibitem{Kozerenko15}
Sergiy Kozerenko.
\newblock On graphs with maximum size in their switching classes.
\newblock {\em Commentationes Mathematicae Universitatis Carolinae},
  56(1):51--61, 2015.

\bibitem{Kratochvil03}
Jan Kratochv{\'i}l.
\newblock {Complexity of Hypergraph Coloring and {S}eidel's Switching}.
\newblock In Hans~L. Bodlaender, editor, {\em {WG}}, volume 2880 of {\em
  {LNCS}}, pages 297--308. Springer Verlag, 2003.

\bibitem{Lindzey14}
Nathan Lindzey.
\newblock {Speeding up Graph Algorithms via Switching Classes}.
\newblock In {\em {Proceedings of IWOCA 2014: 25th International Workshop on
  Combinatorial Algorithms}}, 2015.
\newblock To appear. Preprint available online at
  http://arxiv.org/abs/1408.4900.

\bibitem{MW14}
Ji\v{r}\'{\i} Matou\v{s}ek and Uli Wagner.
\newblock On {G}romov's method of selecting heavily covered points.
\newblock {\em Discrete Comput. Geom.}, 52(1):1--33, July 2014.

\end{thebibliography}

\end{document}